\documentclass[12pt,a4paper,authoryear]{elsarticle}
\usepackage[margin=25mm]{geometry}
\usepackage{amsmath,amssymb,amsthm}
\usepackage{physics}
\usepackage{graphicx}
\usepackage{booktabs}
\usepackage{natbib}
\usepackage[colorlinks=true,linkcolor=blue,citecolor=blue,urlcolor=blue]{hyperref}
\hypersetup{bookmarksdepth=subsubsection}
\newtheorem{lemma}{Lemma}
\newtheorem{proposition}{Proposition}
\theoremstyle{remark}
\newtheorem{remark}{Remark}

\newcommand{\sm}{\mathrm{sm}}
\newcommand{\eq}{\mathrm{eq}}
\newcommand{\sym}{\operatorname{sym}}
\newcommand{\avg}[1]{\left\langle #1 \right\rangle}
\newcommand{\bsigma}{\vb*{\sigma}}
\newcommand{\beps}{\vb*{\varepsilon}}
\newcommand{\bs}{\vb{s}}
\newcommand{\bm}{\vb{m}}
\newcommand{\bA}{\vb{A}}
\newcommand{\bB}{\vb{B}}
\newcommand{\bI}{\vb{I}}
\newcommand{\bT}{\mathbb{T}}
\newcommand{\rt}{\tilde\rho}
\newcommand{\ct}{\tilde\chi}
\newcommand{\tht}{\tilde\theta}
\newcommand{\qt}{\tilde q}
\newcommand{\zt}{\tilde\zeta}
\newcommand{\mut}{\tilde\mu_T}
\newcommand{\tz}{\tilde t_0}
\newcommand{\edot}{\dot\varepsilon}
\newcommand{\gdot}{\dot\gamma}
\begin{document}
\begin{frontmatter}
\title{Averaging in thermodynamic dislocation theory:\\ general macroscopically uniform stress and strain states}
\author[1,2]{K. C. Le\corref{cor}}
\ead{lekhanhchau@tdtu.edu.vn}
\cortext[cor]{Corresponding author.}
\address[1]{Mechanics of Advanced Materials and Structures, Institute for Advanced Study in Technology, Ton Duc Thang University, Ho Chi Minh City, Vietnam}
\address[2]{Faculty of Civil Engineering, Ton Duc Thang University, Ho Chi Minh City, Vietnam}

\begin{abstract}
The averaging procedure of \citet{le2020averaging}, developed there for polycrystalline bars under axially symmetric tension or compression, is extended to arbitrary macroscopically uniform stress and strain states. Starting from the equal probability hypothesis for grain orientations, the mean resolved shear stress and the mean resolved elastic and plastic shear strains are defined as root-mean-square averages over all slip-system orientations. Two exact identities for isotropic orientation averages show that the mean resolved shear stress is proportional to the von Mises equivalent stress, and that the direction of macroscopic plastic flow, obtained from the hypothesis that the plastic slip rate on a system is proportional to the resolved shear stress acting on it, is the stress deviator. The result is an associated $J_2$ flow theory whose hardening law is not fitted but follows from the kinetics of thermally activated dislocation depinning and the evolution equations for the dislocation density and the effective disorder temperature of thermodynamic dislocation theory, written as rates with respect to time so that arbitrary loading paths can be followed. For torsion the theory yields the torque--twist relation of bars and tubes without the classical reductions to a shear stress--strain curve. The parameters for copper are identified jointly from Hopkinson-bar tension, dynamic compression at room and elevated temperatures, and torque--twist records at several twist rates, with the pinning temperature and the reference strain rate fixed by the scaling law for the steady-state flow stress. One set of material parameters, consistent with the earlier compression-only identification, describes tension, compression and torsion from room temperature to $1173$\,K and from $10$ to $2300$\,s$^{-1}$ to $6$\,\% rms over 142 data points; the tension--torsion discrepancy noted by Johnson and Cook is traced to the initial dislocation state of their torsion specimens.
\end{abstract}
\end{frontmatter}

\section{Introduction}

The constitutive description of metal plasticity has developed along four lines that differ in what they take as given and what they claim to predict. It is useful to set them out before stating what the present paper adds, because the theory developed here inherits a definite element from each of them and removes a definite weakness of each.

\noindent\textit{Phenomenological plasticity.} The classical theory \citep{hill1950mathematical,prager1955theory,khan1995continuum} postulates a yield surface, a flow rule and a hardening law, and identifies the functions in them from experiments. In its rate-independent form it is the workhorse of structural analysis, and its viscoplastic extensions \citep{perzyna1966fundamental,bodner1975constitutive,chaboche2008review} add an overstress or a unified viscoplastic potential to cover rate and temperature. Its weakness is that it describes rather than predicts: the hardening function is a fit to the tests from which it was taken, and a new temperature, a new strain rate or a new loading path requires a new identification, often a new functional form. The extensive experimental programs that followed metals over decades of strain rate and hundreds of Kelvin \citep{khan1999behaviors,liang1999critical,rusinek2001shear,khan2004quasi,khan2012variable} document this: the flow stress of the same material at different rates and temperatures is reproduced by successively modified expressions, each with its own set of constants, and the rate sensitivity that emerges is itself a fitted function. Nothing in the framework says why the yield surface should be that of von Mises, why the flow should be associated, or how the hardening at one temperature is related to that at another. What the classical theory has established beyond doubt, and what any dislocation-based theory has to reproduce, is the structure of the macroscopic equations: an isotropic flow potential of the deviator, the Prandtl--Reuss direction of plastic flow, and the associated character of the flow rule for metals without pressure sensitivity.

\noindent\textit{Empirical rate-and-temperature laws.} The models used in impact and forming computations -- \citet{johnson1983constitutive}, \citet{zerilli1987dislocation}, the mechanical threshold stress model of \citet{follansbee1988constitutive}, \citet{preston2003model} -- write the von Mises flow stress as an explicit function of equivalent strain, strain rate and temperature. Some of them, notably \citep{zerilli1987dislocation,follansbee1988constitutive}, motivate their functional forms by thermally activated dislocation glide and by the Kocks--Mecking phenomenology of hardening \citep{mecking1981kinetics,kocks2003physics}; the resulting expressions are nevertheless algebraic fits with five to ten constants per material, and the equivalent plastic strain that appears in them as the hardening variable is a substitute for the state of the material that is valid only for proportional loading at constant rate and temperature. They are calibrated in one stress state, usually compression, and applied in all others by the von Mises hypothesis, which is not examined; \citet{johnson1983constitutive}, who had both tension and torsion data for the same materials, found that the two disagreed and averaged their constants. No account of the observed dependence of the flow stress on the history of the deformation, and no extension to non-proportional paths, is possible within this class.

\noindent\textit{Dislocation-density based theories.} A second generation of models replaced the equivalent plastic strain by the dislocation density as the internal variable, with an evolution equation balancing storage and recovery \citep{mecking1981kinetics,estrin1984unified,bergstrom1970dislocation,kocks2003physics}, and later by several densities distinguished by their location or character \citep{roters2000work,beyerlein2008dislocation,austin2011dislocation}. These theories are physically grounded and predictive within their range: the stage-III hardening and its rate and temperature dependence follow from the storage--recovery balance rather than from a fitted function. They share, however, one structural limitation. The evolution of the dislocation density is written as a kinetic equation with rate constants that are themselves rate- and temperature-dependent, because the recovery term must describe thermally activated rearrangement of a dislocation structure whose state of disorder is not among the variables. The approach to a steady state, the dependence of that state on temperature and rate, and the response to a change of rate or temperature during a test are therefore reproduced by adjusting the recovery function, and not derived. In the language of thermodynamics, the dislocation subsystem has an internal energy but no entropy of its own.

\noindent\textit{Thermodynamic dislocation theory.} \citet{langer2010thermodynamic} supplied the missing variable. The dislocations, as slow configurational degrees of freedom, are assigned an effective (configurational) temperature $\chi$ that measures the disorder of their arrangement, in the same way as the effective temperature of shear transformation zones in amorphous plasticity \citep{falk2011deformation}. The second law applied to this subsystem gives the evolution of $\chi$, driven by the plastic work, toward a steady state $\chi_0$ that is a material property; the dislocation density then relaxes toward the most probable density $\rho_s(\chi)$ at that disorder, and the plastic strain rate follows from thermally activated depinning of dislocations at the Taylor stress. Three consequences distinguish this theory from the dislocation-density models: the steady-state flow stress obeys a scaling law in the single variable $(T/T_P)\ln(\dot\varepsilon_r/\dot\varepsilon)$ that has been confirmed for fcc metals over the whole dynamic range \citep{langer2020scaling}; the approach to the steady state is governed by the entropy production rather than by a recovery function; and the response to changes of rate or temperature during a test, and the onset of shear banding, follow from the same equations \citep{langer2016thermal,langer2017thermal}. The theory was formulated for a single slip system in simple shear. Its application to polycrystals has taken two routes. \citet{lieou2019strain,lieou2020thermodynamic} transported the effective-temperature description to crystal plasticity, with a dislocation density and an effective temperature per slip system and the flow computed by finite elements on a polycrystalline aggregate; this yields the full anisotropic response and the texture evolution, at the cost of a crystal-plasticity computation for every material point and of interaction coefficients between the slip systems that have to be modeled. \citet{langer2015statistical} and, in \citet{le2020averaging}, the present author kept the single-system form and identified the resolved shear stress and strain with effective scalars proportional to the uniaxial stress and strain, the proportionality factors following in \citet{le2020averaging} from an average over the orientations of the slip systems of a randomly oriented polycrystal in axially symmetric tension or compression. That derivation gave the geometric factor $2(1+\nu)$ of the stress-rate equation, but it was tied to the uniaxial state and said nothing about torsion, about combined loading, or about the direction of plastic flow.

\noindent\textit{This paper.} The question addressed here is what replaces the mean resolved shear stress and strains, and what fixes the direction of plastic flow, for a general macroscopically uniform stress and strain state, so that the single-system TDT can be applied to any stress state without a crystal-plasticity computation. The answer turns out to be simple. If the means are defined as root-mean-square averages over the orientation space of slip systems, two exact identities for isotropic fourth-order tensors show that the mean resolved shear stress is proportional to the von Mises equivalent stress for every stress state, and that the hypothesis of a plastic slip rate proportional to the resolved shear stress on each system yields the Prandtl--Reuss flow rule. Thus the equal probability hypothesis for grain orientations and the linear partition of slip produce exactly the associated $J_2$ flow theory of classical plasticity, with no adjustable parameter, and the hardening law of that theory is supplied by TDT. The interpretation of the von Mises stress as a root-mean-square shear stress goes back to \citet{novozhilov1952physical}; the partition of slip over orientations is in the spirit of the slip theory of \citet{batdorf1949mathematical}, but with the rms mean slip rate rather than the individual slips as the quantity whose kinetics is prescribed. Relative to the crystal-plasticity route of \citet{lieou2020thermodynamic} the present theory gives up texture and anisotropy and gains a closed macroscopic model of eight ordinary differential equations per material point, whose structure coincides with that of classical $J_2$ plasticity and can be put into any existing finite-element code in place of a phenomenological hardening law.

The second part of the paper concerns the identification of the material parameters from experiments in different stress states. Torsion tests are of particular value because they deliver, on one specimen, a wide range of strain rates and large strains without necking or barreling. They deliver, however, a torque and a twist angle, and the shear stress--strain curves reported in the literature are reductions of these records under assumptions (thin wall, or the Nadai construction for solid bars) that presuppose a rate-independent, single-valued $\tau(\gamma)$. We therefore derive the torque--twist relation of bars and tubes from the theory, as in \citet{le2018bthermodynamic}, and identify the parameters from the measured torques directly. The identification combines the Hopkinson-bar tension tests and the torsion tests of \citet{johnson1983constitutive}, the high-rate torsion tests of \citet{lindholm1980large}, and the dynamic compression tests of \citet{samanta1971dynamic} at room and elevated temperatures, the pinning temperature and the reference strain rate being fixed beforehand by the scaling law for the steady-state flow stress \citep{le2020two,langer2020scaling}. Quasi-static tests are excluded, since the depinning kinetics is not adequate for them \citep{langer2010thermodynamic}. The outcome is that one set of material parameters describes tension, compression and torsion of copper from room temperature to $1173$\,K and from $10$ to $2300$\,s$^{-1}$ to $6$\,\% rms, that this set coincides, up to the temperature dependence of one rate constant, with the set identified earlier from compression alone, and that the tension--torsion discrepancy noted by Johnson and Cook is a property of the initial dislocation state of their torsion specimens rather than of the material law.

The paper is organized as follows. Section~\ref{sec:averaging} introduces the orientation measure and proves the two averaging identities. Section~\ref{sec:flow} derives the flow rule from the linear partition of slip and shows its variational meaning. Section~\ref{sec:equations} writes the equations of motion of TDT for the general state in time-rate form, including the antisymmetrized depinning law, the temperature-dependent relaxation rate of the effective temperature, adiabatic heating, and the scaling law that fixes the reference quantities. Section~\ref{sec:bvp} completes them to a boundary value problem in terms of the displacement -- balance of momentum, the algebraic equation for the spherical part of the stress, boundary and initial conditions -- shows that its rate form is a linear elastic problem with a known eigenstrain rate, and gives the variational formulation from which the field equations, the natural boundary conditions and the flow rule follow. Section~\ref{sec:uniaxial} shows the reduction to the uniaxial theory of \citet{le2020averaging}, Section~\ref{sec:torsion} derives the torque--twist relation, Section~\ref{sec:identification} presents the joint identification for copper, and Section~\ref{sec:discussion} discusses the results and the next steps.

\section{Averaging over slip-system orientations}
\label{sec:averaging}

\subsection{Orientation measure}

As in \citet{le2020averaging} we consider a representative volume element of a polycrystal containing a large number of randomly oriented grains and assume that all grain orientations are equally probable. A slip system is characterized by the unit normal $\bm$ to the slip plane and the unit slip direction $\bs\perp\bm$. Under the equal probability hypothesis the pair $(\bs,\bm)$ is distributed as $(\vb{Q}\vb{e}_1,\vb{Q}\vb{e}_3)$ with $\vb{Q}$ uniformly distributed (Haar measure) in $SO(3)$. In the spherical coordinates of \citet{le2020averaging}, $\bm=\vb{e}_r(\theta,\varphi)$ and $\bs=\cos\omega\,\vb{e}_\varphi+\sin\omega\,\vb{e}_\theta$, the average of a function $f(\bs,\bm)$ is
\begin{equation}
  \avg{f}=\frac{1}{8\pi^2}\int_0^{2\pi}\!\!\int_0^{2\pi}\!\!\int_0^{\pi} f\,\sin\theta\,d\theta\,d\varphi\,d\omega .
  \label{eq:measure}
\end{equation}
In \citep{le2020averaging} the average was taken over the ``positive'' orientations only, i.e. over the half of orientation space on which the resolved shear stress is positive. For any function that is even under $\bs\to-\bs$ the two averages coincide; all averages used below are of this kind, so we may use the full measure \eqref{eq:measure} throughout.

For a macroscopic (uniform) symmetric tensor $\bA$ we call
\begin{equation}
  A_{\sm}(\bs,\bm)=\bs\vdot\bA\vdot\bm
\end{equation}
its resolved shear component on the system $(\bs,\bm)$. Since $\bs\vdot\bm=0$, only the deviator $\bA'=\bA-\tfrac13(\tr\bA)\bI$ contributes to $A_{\sm}$.

\subsection{Two identities}

The whole theory rests on the isotropic fourth-order tensor
\begin{equation}
  \bT_{ijkl}=\avg{s_i m_j s_k m_l}.
  \label{eq:T}
\end{equation}

\begin{lemma}
$\bT_{ijkl}=\alpha\,\delta_{ij}\delta_{kl}+\beta\,\delta_{ik}\delta_{jl}+\gamma\,\delta_{il}\delta_{jk}$ with $\alpha=\gamma=-\tfrac{1}{30}$, $\beta=\tfrac{2}{15}$.
\end{lemma}

\begin{proof}
By isotropy of the measure, $\bT$ is an isotropic fourth-order tensor and hence of the stated form. The three coefficients follow from three contractions: (i) $\delta_{ij}\bT_{ijkl}=\avg{(\bs\vdot\bm)m_ks_l}=0$ gives $3\alpha+\beta+\gamma=0$; (ii) $\delta_{il}\bT_{ijkl}=\avg{(\bs\vdot\bm)m_js_k}=0$ gives $\alpha+\beta+3\gamma=0$; (iii) $\delta_{ik}\bT_{ijkl}=\avg{(\bs\vdot\bs)m_jm_l}=\avg{m_jm_l}=\tfrac13\delta_{jl}$ gives $\alpha+3\beta+\gamma=\tfrac13$. Solving, $\alpha=\gamma$, $\beta=-4\alpha$, $-10\alpha=\tfrac13$.
\end{proof}

\begin{proposition}[Averaging identities]
For arbitrary symmetric second-order tensors $\bA$, $\bB$,
\begin{align}
  \avg{A_{\sm}B_{\sm}} &= \tfrac{1}{10}\,\bA'\!:\!\bB', \label{eq:id1}\\
  \avg{A_{\sm}\,\sym(\bs\otimes\bm)} &= \tfrac{1}{10}\,\bA'. \label{eq:id2}
\end{align}
\end{proposition}

\begin{proof}
$\avg{A_{\sm}B_{\sm}}=A_{ij}B_{kl}\bT_{ijkl}=\alpha\tr\bA\tr\bB+(\beta+\gamma)\bA\!:\!\bB=\tfrac{1}{10}\bA\!:\!\bB-\tfrac{1}{30}\tr\bA\tr\bB=\tfrac{1}{10}\bA'\!:\!\bB'$. Likewise $A_{ij}\bT_{ijkl}=\alpha(\tr\bA)\delta_{kl}+\beta A_{kl}+\gamma A_{lk}$, whose symmetric part is $-\tfrac{1}{30}(\tr\bA)\bI+\tfrac{1}{10}\bA=\tfrac{1}{10}\bA'$; the term with $s_km_l$ replaced by $m_ks_l$ gives the same by symmetry of $\bA$.
\end{proof}

Identity \eqref{eq:id2} may be read as a reconstruction formula: a deviator is recovered from its resolved shear components on all systems by
\begin{equation}
  \bA'=10\,\avg{A_{\sm}\,\sym(\bs\otimes\bm)}.
  \label{eq:reconstruction}
\end{equation}
The factor 10 is the same in \eqref{eq:id1} and \eqref{eq:id2}: it is the ratio between a macroscopic scalar product and the mean of the products of resolved components, and it will appear once more as the ratio between the macroscopic plastic power and the mean resolved dissipation.

\begin{remark}
Both identities, and all constants below, were also checked numerically by product Gauss--Legendre/trapezoidal quadrature over the orientation space, which is exact for the trigonometric polynomials involved; the coefficients of Lemma 1 are reproduced to $10^{-14}$.
\end{remark}

\subsection{Mean resolved shear stress and strains}

Let $\bsigma$ be the macroscopic Cauchy stress and $\beps=\beps^e+\beps^p$, $\tr\beps^p=0$, the macroscopic small strain. Neglecting, as in \citet{le2020averaging}, the fluctuations of stress and strain inside the grains, the resolved shear stress and strains on the system $(\bs,\bm)$ are $\tau=\sigma_{\sm}$, $\varepsilon^e_{\sm}$ and $\varepsilon^p_{\sm}$. We define their means as root-mean-square (rms) averages:
\begin{equation}
  \bar\tau=\avg{\tau^2}^{1/2},\qquad
  \bar\varepsilon^e=\avg{(\varepsilon^e_{\sm})^2}^{1/2},\qquad
  \dot{\bar\varepsilon}^p=\avg{(\dot\varepsilon^p_{\sm})^2}^{1/2},\qquad
  \bar\varepsilon^p=\int_0^t\dot{\bar\varepsilon}^p\,dt.
  \label{eq:rms}
\end{equation}
The mean plastic strain is defined through its rate, so that it accumulates along non-proportional paths; for proportional loading $\bar\varepsilon^p=\avg{(\varepsilon^p_{\sm})^2}^{1/2}$. By \eqref{eq:id1},
\begin{equation}
  \bar\tau=\sqrt{\tfrac{1}{10}\bsigma'\!:\!\bsigma'}=\sqrt{\tfrac{J_2}{5}}=\frac{\sigma_{\eq}}{\sqrt{15}},\qquad
  \bar\varepsilon^e=\sqrt{\tfrac{1}{10}\beps'^e\!:\!\beps'^e},\qquad
  \dot{\bar\varepsilon}^p=\sqrt{\tfrac{1}{10}\dot\beps^p\!:\!\dot\beps^p}=\sqrt{\tfrac{3}{20}}\,\dot\varepsilon^p_{\eq},
  \label{eq:means}
\end{equation}
where $J_2=\tfrac12\bsigma'\!:\!\bsigma'$, $\sigma_{\eq}=\sqrt{3J_2}$ is the von Mises equivalent stress and $\dot\varepsilon^p_{\eq}=\sqrt{\tfrac23\dot\beps^p\!:\!\dot\beps^p}$ the equivalent plastic strain rate. Thus the mean resolved shear stress of the randomly oriented polycrystal is, exactly and for every stress state, the von Mises stress divided by $\sqrt{15}$. No dependence on the third invariant survives the rms averaging, because an isotropic quadratic form of a deviator is necessarily a multiple of $J_2$.

With the isotropic effective elastic law $\bsigma'=2\mu\,\beps'^e$, $\tr\bsigma=3K\tr\beps$ of the polycrystal, \eqref{eq:means} gives
\begin{equation}
  \bar\tau=2\mu\,\bar\varepsilon^e
  \label{eq:hooke}
\end{equation}
identically, for any stress state. In \citep{le2020averaging} this relation had to be postulated on the basis of homogenization; here it is a consequence of the definitions, because $\tau=2\mu\varepsilon^e_{\sm}$ holds orientation by orientation and the rms average preserves proportionality.

\begin{remark}[Uniaxial state]
For $\bsigma=\operatorname{diag}(0,0,\sigma)$ we get $\sigma_{\eq}=\sigma$ and $\bar\tau=\sigma/\sqrt{15}$; for $\beps^p=\operatorname{diag}(-\tfrac12,-\tfrac12,1)\varepsilon^p_3$, $\bar\varepsilon^p=\sqrt{3/20}\,\varepsilon^p_3$; for $\beps^e=\operatorname{diag}(-\nu,-\nu,1)\varepsilon^e_3$, $\bar\varepsilon^e=(1+\nu)\varepsilon^e_3/\sqrt{15}$. The corresponding averages over positive orientations in \citet{le2020averaging} are $2\sigma/3\pi$, $\varepsilon^p_3/\pi$ and $2(1+\nu)\varepsilon^e_3/3\pi$; the rms values are larger in all three cases by the same factor $(\pi/2)\sqrt{3/5}=\pi\sqrt{3/20}=1.2167$. Substituting $\bar\tau$ and $\bar\varepsilon^e$ into \eqref{eq:hooke} reproduces $\sigma=2\mu(1+\nu)(\varepsilon-\varepsilon^p)$; the geometric factor $2(1+\nu)$ is unchanged.
\end{remark}

\begin{remark}[Relation to classical results]
\citet{novozhilov1952physical} showed that the mean square of the shear traction $\vb{t}_s=\bsigma\bm-(\bm\vdot\bsigma\vdot\bm)\bm$ over all planes through a point is $\tau_i^2:=\avg{|\vb{t}_s|^2}_{\bm}=\tfrac{6}{15}J_2$, so that the rms shear stress over all planes, $\tau_i=\sqrt{6/15}\sqrt{J_2}=\sqrt{2/15}\,\sigma_{\eq}$, is proportional to the von Mises stress. Identity \eqref{eq:means}$_1$ follows from this by one more averaging: on a given plane the resolved shear stress is $\tau=\vb{t}_s\vdot\bs$ with $\bs$ uniformly distributed on the unit circle of the plane, hence $\avg{\tau^2}_{\bs}=\tfrac12|\vb{t}_s|^2$ and $\avg{\tau^2}=\tfrac12\tau_i^2=\tfrac12\cdot\tfrac{6}{15}J_2=\tfrac15J_2$, i.e. $\bar\tau=\tau_i/\sqrt2$. Conversely, Novozhilov's constant is contained in Lemma 1: $\tau_i^2=|\bsigma\bm|^2-(\bm\vdot\bsigma\vdot\bm)^2=(\tfrac{10}{15}-\tfrac{4}{15})J_2$. The mean over positive orientations of \citet{le2020averaging}, $\avg{|\tau|}$, is by contrast a $J_2$--$J_3$ function; its Lode-angle dependence is weak (about 2.7\,\% between uniaxial tension and pure shear at equal $J_2$, see Remark~\ref{rem:family}) but not zero.
\end{remark}

\section{Direction of plastic flow}
\label{sec:flow}

The averaging of Section~\ref{sec:averaging} fixes the scalar variables $\bar\tau$, $\bar\varepsilon^e$, $\bar\varepsilon^p$ on which the one-dimensional TDT of \citet{langer2010thermodynamic,le2020averaging} operates, but not the tensor $\dot\beps^p$. In the uniaxial case its direction was fixed by symmetry; in general an additional hypothesis on how plastic slip is distributed over orientations is needed.

\noindent\textit{Hypothesis (linear partition of slip).} At each instant the plastic resolved shear strain rate on a system is proportional to the resolved shear stress acting on it:
\begin{equation}
  \dot\varepsilon^p_{\sm}(\bs,\bm)=\frac{\dot{\bar\varepsilon}^p}{\bar\tau}\,\tau(\bs,\bm).
  \label{eq:partition}
\end{equation}
The normalization is not an extra assumption: taking the rms of both sides of \eqref{eq:partition} gives an identity by \eqref{eq:rms}. The mean plastic slip rate $\dot{\bar\gamma}^p=2\dot{\bar\varepsilon}^p$ (engineering shear) is the quantity whose kinetics TDT provides; \eqref{eq:partition} says that the system on which $\tau=\bar\tau$ slips at exactly this rate, systems under higher resolved stress faster, systems under lower resolved stress slower, and systems with $\tau<0$ in the opposite sense.

Applying the reconstruction formula \eqref{eq:reconstruction} to $\dot\beps^p$ and using \eqref{eq:id2} once more,
\begin{equation}
  \dot\beps^p=10\,\avg{\dot\varepsilon^p_{\sm}\sym(\bs\otimes\bm)}=\frac{10\,\dot{\bar\varepsilon}^p}{\bar\tau}\avg{\tau\,\sym(\bs\otimes\bm)}=\frac{\dot{\bar\varepsilon}^p}{\bar\tau}\,\bsigma'.
  \label{eq:PR}
\end{equation}
This is the Prandtl--Reuss flow rule. In terms of the von Mises variables \eqref{eq:means} it takes the standard form
\begin{equation}
  \dot\beps^p=\frac{3}{2}\frac{\dot\varepsilon^p_{\eq}}{\sigma_{\eq}}\,\bsigma',\qquad
  \dot\beps^p=\dot\lambda\,\pdv{\bar\tau}{\bsigma},\quad \dot\lambda=10\,\dot{\bar\varepsilon}^p,
  \label{eq:flowrule}
\end{equation}
i.e. the flow is associated with the potential $\bar\tau(\bsigma)$ and the direction of plastic flow is the direction of the stress deviator. Thus the equal probability hypothesis together with the linear partition \eqref{eq:partition} yields exactly the von Mises/Prandtl--Reuss structure, with no adjustable parameter.

\noindent\textit{Plastic power.} By \eqref{eq:id1} and \eqref{eq:partition},
\begin{equation}
  \bsigma\!:\!\dot\beps^p=10\,\avg{\tau\,\dot\varepsilon^p_{\sm}}=10\,\bar\tau\,\dot{\bar\varepsilon}^p=5\,\bar\tau\,\dot{\bar\gamma}^p=\sigma_{\eq}\,\dot\varepsilon^p_{\eq}.
  \label{eq:power}
\end{equation}
The macroscopic plastic power equals ten times the mean of the resolved dissipation $\tau\dot\varepsilon^p_{\sm}$, with the same factor 10 as in \eqref{eq:reconstruction}. It is this quantity that drives the evolution of the dislocation density and of the effective temperature in Section~\ref{sec:equations}. In \citep{le2020averaging} the analogous relation, $\sigma\dot\varepsilon^p=\tfrac{3}{2}\pi^2\bar\tau\dot{\bar\varepsilon}^p$, holds only for the uniaxial state.

\begin{proposition}[Variational characterization]
Among all distributions $\dot\varepsilon^p_{\sm}(\bs,\bm)$ of resolved plastic strain rates with prescribed rms value $\dot{\bar\varepsilon}^p$, the linear partition \eqref{eq:partition} is the unique one maximizing the mean resolved dissipation $\avg{\tau\dot\varepsilon^p_{\sm}}$, and hence the macroscopic plastic power $\bsigma\!:\!\dot\beps^p$.
\end{proposition}

\begin{proof}
By the Cauchy--Schwarz inequality in $L^2$ of the orientation measure,
\[\avg{\tau\dot\varepsilon^p_{\sm}}\le\avg{\tau^2}^{1/2}\avg{(\dot\varepsilon^p_{\sm})^2}^{1/2}=\bar\tau\dot{\bar\varepsilon}^p,\]
with equality if and only if $\dot\varepsilon^p_{\sm}\propto\tau$ with a positive factor.
\end{proof}

Hypothesis \eqref{eq:partition} is therefore the principle of maximum plastic dissipation transported to orientation space, with the rms mean slip rate as the constrained quantity. This is the reason why the resulting macroscopic flow rule is associated.

\begin{remark}[A family of partitions]
\label{rem:family}
Hypothesis \eqref{eq:partition} is the member $n=1$ of the family $\dot\varepsilon^p_{\sm}\propto\operatorname{sgn}(\tau)|\tau|^n$. Each member gives an associated flow rule with the isotropic potential $\Phi_n(\bsigma)=\avg{|\tau|^{n+1}}^{1/(n+1)}$, an $L^{n+1}$ norm of the resolved shear stress over orientation space. The member $n=0$ (all positively oriented systems slip at the same rate) is the hypothesis underlying the averages over positive orientations in \citet{le2020averaging}: $\Phi_0=\avg{|\tau|}$ equals $2\sigma/3\pi$ in uniaxial tension. The limit $n\to\infty$ gives $\Phi_\infty=\max|\tau|=\tfrac12(\sigma_1-\sigma_3)$, i.e. Tresca. Numerically, the ratio $\Phi_n/\sqrt{J_2}$ in pure shear relative to uniaxial tension is 1.027 for $n=0$, 1.000 for $n=1$, 0.996 for $n=2$, 1.000 for $n=3$, 1.007 for $n=4$, 1.014 for $n=5$, 1.029 for $n=7$, and $2/\sqrt3=1.155$ for $n=\infty$. The randomly oriented polycrystal is thus close to von Mises for every reasonable partition. Exactly $J_2$ are the members $n=1$ and $n=3$: $\avg{\tau^2}$ and $\avg{\tau^4}$ are isotropic polynomial invariants of the deviator of degrees 2 and 4, and the only such invariants are $J_2$ and $J_2^2$ ($J_3$ being of degree 3); from degree 6 on, $J_3^2$ enters. Among these, $n=1$ is singled out by the rms definitions \eqref{eq:rms} and by the maximum dissipation property of Proposition 2.
\end{remark}

\section{Equations of motion}
\label{sec:equations}

\subsection{Kinetics of depinning}

As in \citet{le2020averaging} we assume that the mean plastic slip rate is governed by the kinetics of thermally activated dislocation depinning \citep{langer2010thermodynamic},
\begin{equation}
  \dot{\bar\gamma}^p\,t_0=b\sqrt\rho\,\Big[\exp\Big(-\frac{T_P}{T}e^{-\bar\tau/\tau_T}\Big)-\exp\Big(-\frac{T_P}{T}e^{\bar\tau/\tau_T}\Big)\Big],\qquad \tau_T=\mu_T b\sqrt\rho,
  \label{eq:depinning}
\end{equation}
with $t_0$ a microscopic time, $b$ the Burgers vector, $\rho$ the dislocation density, $T$ the ordinary temperature, $k_BT_P$ the pinning energy barrier and $\tau_T$ the Taylor stress. The second term in the bracket is the rate of backward depinning; it makes the slip rate an odd function of the stress that vanishes at $\bar\tau=0$, and it is negligible against the first whenever $e^{-\bar\tau/\tau_T}\ll1$ or $T/T_P\ll1$, which is the case in all room-temperature applications of the theory. At the temperatures of the hot compression tests below, $T/T_P\approx0.02$--$0.03$, the forward term alone would give a slip rate $b\sqrt\rho\,e^{-T_P/T}/t_0$ at zero stress comparable to the imposed strain rates, and no steady state would exist at low stress; the antisymmetrized form removes this artifact of the original law without changing any result in the regime where the original law was used.

Substituting $\bar\tau=\sigma_{\eq}/\sqrt{15}$ and $\dot{\bar\gamma}^p=2\dot{\bar\varepsilon}^p=\sqrt{3/5}\,\dot\varepsilon^p_{\eq}$ from \eqref{eq:means},
\begin{equation}
  \begin{aligned}
  &\dot\varepsilon^p_{\eq}\,t_0=\sqrt{\tfrac53}\,q(\sigma_{\eq},\rho,T),\qquad \bar\mu_T=\sqrt{15}\,\mu_T,\\
  &q=b\sqrt\rho\,\Big[\exp\Big(-\frac{T_P}{T}e^{-\sigma_{\eq}/\bar\mu_Tb\sqrt\rho}\Big)-\exp\Big(-\frac{T_P}{T}e^{\sigma_{\eq}/\bar\mu_Tb\sqrt\rho}\Big)\Big],
  \end{aligned}
  \label{eq:q}
\end{equation}
Equation \eqref{eq:q} replaces Eq.~(5) of \citet{le2020averaging}, where the constants were $\pi/2$ and $3\pi/2$ instead of $\sqrt{5/3}$ and $\sqrt{15}$.

\subsection{Stress rate}

From $\beps=\beps^e+\beps^p$, the isotropic elastic law and \eqref{eq:flowrule},
\begin{equation}
  \dot\bsigma=K(\tr\dot\beps)\bI+2\mu\Big(\dot\beps'-\frac{3}{2}\frac{\dot\varepsilon^p_{\eq}}{\sigma_{\eq}}\bsigma'\Big),\qquad
  \dot\varepsilon^p_{\eq}=\sqrt{\tfrac53}\,\frac{q(\sigma_{\eq},\rho,T)}{t_0}.
  \label{eq:stressrate}
\end{equation}
This generalizes Eq.~(6) of \citet{le2020averaging}. We keep the rates with respect to time. The form with the strain as independent variable used in \citet{le2020averaging} is available only for proportional loading at constant strain rate; for a general path $\beps(t)$, e.g. tension followed by torsion, the direction of $\dot\beps'$ and the magnitude of the rate enter \eqref{eq:stressrate} separately and no single scalar strain can serve as time-like variable.

\subsection{Dislocation density and effective temperature}

The evolution equations for $\rho$ and for the effective disorder temperature $\chi$ of \citet{langer2010thermodynamic} are driven by the plastic power. With \eqref{eq:power} they read
\begin{align}
  \dot\rho &=\frac{\kappa_\rho}{\mu_T\,\zeta^2(\rho,q,T)\,b^2}\,\sigma_{\eq}\dot\varepsilon^p_{\eq}\Big[1-\frac{\rho}{\rho_s(\chi)}\Big],\qquad \rho_s(\chi)=\frac{1}{a^2}e^{-e_d/\chi}, \label{eq:rhodot}\\
  \dot\chi &=\frac{\kappa_\chi(T)}{\mu_T}\,\sigma_{\eq}\dot\varepsilon^p_{\eq}\Big(1-\frac{\chi}{\chi_0}\Big),\qquad
  \zeta(\rho,q,T)=\ln\frac{T_P}{T}-\ln\ln\frac{b\sqrt\rho}{aq}. \label{eq:chidot}
\end{align}
Three remarks on the comparison with Eq.~(8) of \citet{le2020averaging}. First, the driving term $\sigma\dot\varepsilon^p$ (there written as $\sigma q/\bar q_0$ per unit strain) is replaced by the plastic power $\bsigma\!:\!\dot\beps^p=\sigma_{\eq}\dot\varepsilon^p_{\eq}$. Second, the argument of $\zeta$ is the current dimensionless plastic slip rate $q=\dot{\bar\gamma}^pt_0$, as in the original formulation of \citet{langer2010thermodynamic}; in \citet{le2020averaging} it was replaced by the constant $\bar q_0$ of the constant-rate test, which is its value after yield, when the plastic rate has caught up with the imposed rate. For a variable or non-proportional loading the current rate must be retained. With the current rate, $\zeta$ has moreover a closed form: inverting the depinning law \eqref{eq:depinning} (with the backward term neglected), $\ln\ln(b\sqrt\rho/aq)=\ln(T_P/T)-\bar\tau/\tau_T$ up to the factor $a/b$ absorbed in the definition of $q$, hence
\begin{equation}
  \zeta(\rho,q,T)=\frac{\bar\tau}{\tau_T}=\frac{\sigma_{\eq}}{\bar\mu_Tb\sqrt\rho},
  \label{eq:zeta}
\end{equation}
which is the original meaning of this quantity in \citet{langer2010thermodynamic}, the ratio of the driving stress to the Taylor stress. The $\rho$ equation therefore contains no logarithms at all:
\begin{equation}
  \dot\rho=\kappa_\rho\,\frac{\bar\mu_T^2\,\rho\,\dot\varepsilon^p_{\eq}}{\mu_T\,\sigma_{\eq}}\Big[1-\frac{\rho}{\rho_s(\chi)}\Big],
  \label{eq:rhodot2}
\end{equation}
i.e. the storage rate is proportional to $\rho$ and to the plastic strain rate and inversely proportional to the flow stress. In a constant-rate test the form \eqref{eq:zeta} and the form with $\bar q_0$ of \citet{le2020averaging} give stress--strain curves that differ by less than 0.1\,MPa. Third, the rate constant $\kappa_\chi$ of the effective temperature is allowed to depend on the ordinary temperature. Following \citep{le2020two} we take
\begin{equation}
  K_\chi(T)=c_0\exp(T/c_1),
  \label{eq:KchiT}
\end{equation}
with material constants $c_0$, $c_1$, while $\kappa_\rho$ is independent of temperature and strain rate. This dependence expresses that the approach of the dislocation structure to its steady state is thermally assisted; at the temperatures of the hot compression tests it makes the flow stress saturate within the first tenth of strain, as observed, whereas a constant $K_\chi$ fitted to the room-temperature data gives a much slower approach.

\subsection{Dimensionless form}

With $\rt=a^2\rho$, $\ct=\chi/e_d$, $\tht=T/T_P$, $\rt_s=e^{-1/\ct}$ and, as in \citet{le2020averaging}, the assumption that the dimensionless Taylor modulus $\mut$ scales with $\mu(T)$ through a material constant $s$,
\begin{equation}
  \begin{aligned}
  &\qt(\sigma_{\eq},\rt,\tht)=\sqrt{\rt}\,\Big[\exp\Big(-\tfrac{1}{\tht}e^{-x}\Big)-\exp\Big(-\tfrac{1}{\tht}e^{x}\Big)\Big],\qquad x=\frac{\sigma_{\eq}}{\mut\sqrt{\rt}}=\zt,\\
  &\mut=\frac{b}{a}\bar\mu_T=\frac{b}{a}\sqrt{15}\,\mu_T=:s\,\mu(T),
  \end{aligned}
  \label{eq:qtilde}
\end{equation}
the plastic strain rate is $\dot\varepsilon^p_{\eq}=\qt/\tz$ with the single time constant
\begin{equation}
  \tz=\sqrt{\tfrac35}\,\frac{a}{b}\,t_0,
  \label{eq:t0tilde}
\end{equation}
the counterpart of the convention $(2a/\pi b)t_0$ of \citet{le2020averaging}. The system \eqref{eq:stressrate}--\eqref{eq:chidot} becomes
\begin{equation}
  \begin{aligned}
    \dot\bsigma &=K(\tr\dot\beps)\bI+2\mu\Big(\dot\beps'-\frac{3}{2}\frac{\qt}{\tz\,\sigma_{\eq}}\bsigma'\Big),\\
    \dot\rt &=K_\rho\,\frac{\mut\,\rt\,\qt}{\sigma_{\eq}\,\tz}\Big[1-\frac{\rt}{\rt_s(\ct)}\Big],\\
    \dot\ct &=K_\chi(T)\,\frac{\sigma_{\eq}\,\qt}{\mut\,\tz}\Big(1-\frac{\ct}{\ct_0}\Big),
  \end{aligned}
  \label{eq:system}
\end{equation}
with $K_\rho=\kappa_\rho a/b$ and $K_\chi=\kappa_\chi b/(ae_d)$ as in \citet{le2020averaging}. For a prescribed strain path $\beps(t)$, \eqref{eq:system} is a closed system of $6+2$ ordinary differential equations for $\bsigma$, $\rt$, $\ct$. For a prescribed stress path $\bsigma(t)$ the first equation is solved for $\dot\beps$ instead. Mixed control, as in the tension--torsion test of a thin-walled tube where $\varepsilon_{zz}(t)$ and $\gamma_{z\theta}(t)$ are prescribed and the remaining stress components vanish, is handled by solving the stress-rate equation for the unprescribed strain rates. Dividing \eqref{eq:system} by a constant $\dot\varepsilon_{\eq}$ for proportional loading recovers the strain-parametrized form with $\qt_0=\dot\varepsilon_{\eq}\tz$ of \citet{le2020averaging}.

\subsection{Finite strains}
\label{sec:finite}

The equations have been written for small strains, with the additive decomposition $\beps=\beps^e+\beps^p$, but they are applied below to plastic strains of order one. The justification is the smallness of the \emph{elastic} strain in metals and the isotropy of the polycrystal, and it is worth stating precisely, since the torsion tests reach $\gamma=3$.

The finite-strain kinematics of an elastic-plastic solid is the multiplicative decomposition $\vb{F}=\vb{F}^e\vb{F}^p$ of the deformation gradient into a plastic part, which carries the material through the intermediate configuration produced by dislocation glide, and an elastic part, which distorts and rotates the lattice \citep{kroner1959allgemeine,lee1969elastic}. For an isotropic elastic response and an isotropic flow rule the theory takes a canonical form \citep{simo1988aframework,simo1992algorithms,miehe1992canonical,gurtin2005theory,gurtin2010mechanics}: the plastic spin can be set to zero without loss of generality (the material is plastically irrotational), and the rate of deformation $\vb{D}=\sym(\dot{\vb{F}}\vb{F}^{-1})$ decomposes additively,
\begin{equation}
  \vb{D}=\vb{D}^e+\vb{D}^p,\qquad \vb{D}^p=\dot\varepsilon^p_{\eq}\,\frac{3}{2}\frac{\bsigma'}{\sigma_{\eq}},
  \label{eq:Dsplit}
\end{equation}
with the Cauchy stress $\bsigma$ related to the elastic logarithmic strain by the isotropic law and its objective (corotational) rate to $\vb{D}^e$ by the elastic moduli, up to terms of the order of the elastic strain \citep{simo1992algorithms,holzapfel2000nonlinear}. The system \eqref{eq:system} is to be read in this sense: $\dot\beps$ stands for $\vb{D}$, $\dot\bsigma$ for an objective rate of the Cauchy stress, and the flow rule \eqref{eq:flowrule} and the plastic power \eqref{eq:power} are exactly the finite-strain expressions \eqref{eq:Dsplit} and $\bsigma\!:\!\vb{D}^p$. The internal variables $\rt$, $\ct$, $T$ are scalars and need no transport. The orientation average of Section~\ref{sec:averaging} is unaffected, because the resolved quantities on a slip system are defined from $\bsigma$ and $\vb{D}$ in the current configuration and the equal probability hypothesis is a statement about the current distribution of orientations; what the finite-strain theory adds, and the present isotropic version omits, is the evolution of that distribution, i.e. texture.

For the uniform tests used below the reading is explicit. In uniaxial tension and compression at constant true strain rate, $\vb{D}$ is diagonal and constant in direction, the objective rate coincides with the material rate, and the equations of Section~\ref{sec:uniaxial} hold verbatim with $\varepsilon$ the logarithmic strain and $\sigma$ the true stress, which are the quantities reported by \citet{samanta1971dynamic} and \citet{johnson1983constitutive}. In torsion of a tube or bar the deformation at each radius is a simple shear of amount $\gamma=\kappa r$; $\vb{D}$ has the single component $D_{z\theta}=\gdot/2$, so that the flow rule keeps $\vb{D}^p$ a pure shear, and the difference between the objective and the material rate of $\bsigma$ is the term $\vb{W}\bsigma-\bsigma\vb{W}$ with the spin $W_{z\theta}=\gdot/2$, which generates normal stresses of order $\tau\gamma$ times the elastic strain $\tau/\mu$ -- about $1$\,\% of $\tau$ at $\gamma=3$ for the stresses of copper -- and modifies the shear stress by the same relative amount. This is the Swift effect of the finite-strain $J_2$ theory \citep{simo1988aframework}; at the level of the torque it is below the digitization accuracy and is neglected. The strain measure $\gamma=R_m\varphi/L$ of the thin-walled specimen is exact for simple shear at any twist. The theory is therefore applicable up to the strains of the torsion tests without change, with the one reservation that texture, which develops in copper beyond $\gamma\approx1$--$2$ and is the mechanism behind stage-IV hardening, is outside it.

\subsection{Adiabatic heating}

For tests shorter than about $0.1$\,s the heat generated by plastic work has no time to leave the specimen and the temperature rises adiabatically,
\begin{equation}
  \rho_m c\,\dot T=\beta\,\sigma_{\eq}\dot\varepsilon^p_{\eq},
  \label{eq:adiabatic}
\end{equation}
with the Taylor--Quinney fraction $\beta$, the mass density $\rho_m$ and the specific heat $c$; $\mu(T)$ and $\tht=T/T_P$ in \eqref{eq:system} are evaluated at the current temperature. Longer tests are treated as isothermal. For copper ($\rho_mc=3.45$\,MJ\,m$^{-3}$K$^{-1}$, $\beta=0.9$) the rise is $18$\,K at $\varepsilon=0.3$ and $451$\,s$^{-1}$, which lowers the stress by about $9$\,MPa, and $12$--$20$\,K at $\varepsilon=0.5$ in the hot compression tests at $10^3$\,s$^{-1}$.

\subsection{Scaling law and reference quantities}
\label{sec:scaling}

In the steady state of a constant-rate test, $\rt=\rt_s(\ct_0)$ and the depinning law \eqref{eq:qtilde} can be inverted for the flow stress. With the backward term neglected,
\begin{equation}
  \frac{\sigma_s}{\mu(T)\,s\sqrt{\rt_s}}=\ln\frac{1}{x},\qquad x=\frac{T}{T_P}\ln\frac{\edot_r}{\edot},\qquad
  \edot_r=\frac{\sqrt{\rt_s}}{\tz},
  \label{eq:scalinglaw}
\end{equation}
which is the double-logarithmic scaling law for the steady-state flow stress of fcc metals \citep{le2020two}. It contains three combinations of the material parameters: the steady-state Taylor stress over the shear modulus, $s_T:=s\sqrt{\rt_s}$, the pinning temperature $T_P$, and the reference strain rate $\edot_r$. These three quantities were identified in \citet{le2020two} from the steady-state stresses of Samanta's dynamic compression tests of copper \citep{samanta1971dynamic} at three temperatures and four strain rates, with the result
\begin{equation}
  s_T=6.3915\times10^{-3},\qquad T_P=45\,000\ \text{K},\qquad \edot_r=3.16\times10^{12}\ \text{s}^{-1}.
  \label{eq:MRCvalues}
\end{equation}
We keep these values fixed. Then $T_P$ is no longer a free parameter, and $s$ and $\tz$ follow from $\ct_0$,
\begin{equation}
  s=\frac{s_T}{\sqrt{\rt_s(\ct_0)}}=s_T\,e^{1/2\ct_0},\qquad \tz=\frac{\sqrt{\rt_s(\ct_0)}}{\edot_r}=\frac{e^{-1/2\ct_0}}{\edot_r}.
  \label{eq:derived}
\end{equation}
This replaces the convention $\tz=10^{-12}$\,s of \citet{le2020averaging}. The distinction matters: the steady-state stress depends on $T_P$ only through $T_P/\ln(\edot_r/\edot)$, so a pinning temperature identified with an arbitrary $\tz$ is not comparable with one identified with the reference rate of the scaling law. The remaining material parameters are $\ct_0$, $K_\rho$, $c_0$ and $c_1$. For $\ct_0$ we adopt the bound $\ct_0\le0.25$ of \citet{langer2010thermodynamic}, which all previous identifications have respected.

\section{Boundary value problem and variational formulation}
\label{sec:bvp}

Section~\ref{sec:equations} gives the constitutive equations at a material point. For a body they must be completed by the balance of momentum, by the equation for the spherical part of the stress, which \eqref{eq:system} contains only in rate form, and by boundary and initial conditions. We state the resulting boundary value problem in terms of the displacement, show that its rate form is a linear elastic problem with a known eigenstrain rate, and give the variational formulation from which the field equations and the natural boundary conditions follow, in the manner of \citet{le2018dthermodynamic}.

\subsection{Field equations}

Let the body occupy the region $\Omega$ with boundary $\partial\Omega=\partial_u\Omega\cup\partial_t\Omega$, and let $\vb{u}(\vb{x},t)$ be the displacement, $\beps=\sym\grad\vb{u}$ the strain and $\beps^p$ the plastic strain, $\tr\beps^p=0$ by the flow rule \eqref{eq:flowrule}. The unknown fields are $\vb{u}$ (3), $\bsigma$ (6), $\rt$, $\ct$ and $T$, twelve in all. The equations are:

\noindent\textit{Balance of momentum.}
\begin{equation}
  \div\bsigma+\rho_m\vb{b}=\rho_m\ddot{\vb{u}}\quad\text{in }\Omega,
  \label{eq:momentum}
\end{equation}
with $\vb{b}$ the body force. For the tests of Section~\ref{sec:identification} the inertial term is dropped: in a Hopkinson-bar test the stress equilibrates across the specimen within a few transit times of the elastic wave, about $1\,\mu$s for a millimeter-sized specimen, while the test lasts several hundred microseconds, so that the specimen is in equilibrium at every instant, $\div\bsigma+\rho_m\vb{b}=\vb{0}$. The inertial term must be kept in wave-propagation and impact problems, for which the theory is equally suited.

\noindent\textit{Spherical part of the stress.} Taking the trace of the first equation of \eqref{eq:system} gives $\tr\dot\bsigma=3K\tr\dot\beps$, since $\tr\bsigma'=0$. With $\bsigma=\sigma_m\bI+\bsigma'$, $\sigma_m=\tr\bsigma/3$, and a stress-free reference state this integrates to
\begin{equation}
  \sigma_m=K\,\tr\beps=K\,\div\vb{u}.
  \label{eq:spherical}
\end{equation}
The volumetric response is elastic at all times: plastic flow is isochoric and does not affect the pressure. Equation \eqref{eq:spherical} is algebraic; only the deviator evolves,
\begin{equation}
  \dot\bsigma'=2\mu\Big(\dot\beps'-\dot\beps^p\Big),\qquad
  \dot\beps^p=\frac32\,\frac{\qt(\sigma_{\eq},\rt,\tht)}{\tz\,\sigma_{\eq}}\,\bsigma',
  \label{eq:deviator}
\end{equation}
together with the second and third equations of \eqref{eq:system} for $\rt$ and $\ct$.

\noindent\textit{Heat conduction.} The general form of \eqref{eq:adiabatic} is
\begin{equation}
  \rho_mc\,\dot T=\div(k\grad T)+\beta\,\sigma_{\eq}\dot\varepsilon^p_{\eq},
  \label{eq:heat}
\end{equation}
with the thermal conductivity $k$; the adiabatic and isothermal cases of Section~\ref{sec:equations} are its limits for short and long tests. The thermoelastic coupling is neglected.

\noindent\textit{Boundary and initial conditions.}
\begin{equation}
  \vb{u}=\bar{\vb{u}}\ \text{on }\partial_u\Omega,\qquad \bsigma\vb{n}=\bar{\vb{t}}\ \text{on }\partial_t\Omega,
  \label{eq:bc}
\end{equation}
with $\vb{n}$ the outward normal, plus a temperature or heat-flux condition for \eqref{eq:heat}, and at $t=0$
\begin{equation}
  \vb{u}=\vb{u}_0,\quad \dot{\vb{u}}=\vb{v}_0,\quad \bsigma'=\bsigma'_0,\quad \rt=\rt_i,\quad \ct=\ct_i,\quad T=T_0\quad\text{in }\Omega,
  \label{eq:ic}
\end{equation}
where $\vb{u}_0$ and $\vb{v}_0$ are needed only in the dynamic case. The initial dislocation state $(\rt_i,\ct_i)$ may vary in space, for instance in a specimen with a cold-worked surface layer; for the tests below it is uniform and is one of the unknowns of the identification.

The count is complete: the three components of \eqref{eq:momentum}, the one equation \eqref{eq:spherical}, the five of \eqref{eq:deviator}, the two internal-variable equations and \eqref{eq:heat} are twelve equations for the twelve fields. The system is of first order in time for $(\bsigma',\rt,\ct,T)$ and of second order (or zeroth, in the quasi-static case) for $\vb{u}$.

\subsection{Displacement formulation and the rate problem}
\label{sec:rateproblem}

Consider the quasi-static case. Suppose the state $(\bsigma',\rt,\ct,T)$ is known at the instant $t$. Then the plastic strain rate \eqref{eq:deviator}$_2$ is a known field, since it depends on the current stress and internal variables only and not on the velocity $\vb{v}=\dot{\vb{u}}$. This is the decisive difference from rate-independent plasticity: there is no yield surface, no distinction between loading and unloading and no consistency condition to be enforced; the elastic and the plastic regime are distinguished only by the size of $\qt$, which is exponentially small below the Taylor stress. Differentiating \eqref{eq:momentum}, \eqref{eq:spherical}, \eqref{eq:bc} with respect to time and inserting \eqref{eq:deviator} yields the rate boundary value problem for the velocity,
\begin{equation}
  \begin{aligned}
    &\div\big[K(\div\vb{v})\bI+2\mu(\sym\grad\vb{v})'\big]=\div\big(2\mu\dot\beps^p\big)-\rho_m\dot{\vb{b}}\quad\text{in }\Omega,\\
    &\vb{v}=\dot{\bar{\vb{u}}}\ \text{on }\partial_u\Omega,\qquad
    \big[K(\div\vb{v})\bI+2\mu\big((\sym\grad\vb{v})'-\dot\beps^p\big)\big]\vb{n}=\dot{\bar{\vb{t}}}\ \text{on }\partial_t\Omega.
  \end{aligned}
  \label{eq:ratebvp}
\end{equation}
These are the Navier equations of isotropic linear elasticity for $\vb{v}$, with the plastic strain rate playing the role of a prescribed eigenstrain rate in the sense of \citet{eshelby1957determination} and \citet{mura2013micromechanics}. For $K>0$, $\mu>0$ and a part $\partial_u\Omega$ of positive area the problem has exactly one solution $\vb{v}\in H^1(\Omega)$ by Korn's inequality and the Lax--Milgram lemma \citep{gurtin1973linear,horgan1995korns}; if $\partial_u\Omega$ is empty the solution is unique up to a rigid velocity, provided the loads are self-equilibrated. Equivalently, $\vb{v}$ is the unique minimizer of the strictly convex rate functional
\begin{equation}
  \mathcal{J}[\vb{v}]=\int_\Omega\Big[\tfrac12K(\div\vb{v})^2+\mu\,\big|(\sym\grad\vb{v})'-\dot\beps^p\big|^2\Big]dV-\int_\Omega\rho_m\dot{\vb{b}}\cdot\vb{v}\,dV-\int_{\partial_t\Omega}\dot{\bar{\vb{t}}}\cdot\vb{v}\,dS
  \label{eq:ratefunctional}
\end{equation}
among the velocity fields with $\vb{v}=\dot{\bar{\vb{u}}}$ on $\partial_u\Omega$, whose Euler--Lagrange equations and natural boundary conditions are \eqref{eq:ratebvp}. This is the rate principle of \citet{hill1958general} specialized to a material without a loading--unloading switch, for which it is quadratic.

Once $\vb{v}$ is known, \eqref{eq:deviator} gives $\dot\bsigma'$ and the internal-variable and heat equations give $\dot\rt$, $\dot\ct$, $\dot T$. The complete problem is therefore an evolution equation in time for the state fields $(\bsigma',\rt,\ct,T)$, whose right-hand side contains the solution of the linear elliptic problem \eqref{eq:ratebvp} at every instant, with $\vb{u}$ and $\sigma_m$ recovered from $\vb{v}$ and \eqref{eq:spherical}. This is the structure of the elastic-viscoplastic problem of \citet{perzyna1966fundamental} and \citet{duvaut1976inequalities}, for which the existence theory and the numerical treatment are classical \citep{simo1998computational}; what distinguishes TDT within this class is the form of the plastic strain rate function, which is the depinning law \eqref{eq:qtilde} rather than an overstress power law, and the two internal variables with their own evolution. The right-hand side is smooth in the state: $\qt$ is an odd analytic function of $\sigma_{\eq}$, so that the quotient $\qt/\sigma_{\eq}$ in \eqref{eq:deviator} and \eqref{eq:system} has a finite limit at $\sigma_{\eq}=0$, and $\rt_s(\ct)$ is smooth for $\ct>0$. After spatial discretization by finite elements the problem is a system of ordinary differential equations with a locally Lipschitz right-hand side, which has a unique solution by the Picard--Lindel\"of theorem and is integrated by the stiff solvers used in Sections~\ref{sec:torsion} and \ref{sec:identification}. Two consequences are worth stating. First, since the rate problem is linear elastic, ellipticity is never lost and no bifurcation occurs at the level of the rate equations; localization, such as adiabatic shear banding, develops through the evolution of the state fields, in particular through thermal softening via $\mu(T)$ and $\tht=T/T_P$ in \eqref{eq:heat}, and its width is set by the rate sensitivity and by heat conduction, as in rate-dependent theories generally \citep{needleman1988material}. Second, the uniform tests of Sections~\ref{sec:uniaxial}--\ref{sec:torsion} are the cases in which \eqref{eq:ratebvp} is solved by inspection: in tension the stress is uniform, and in torsion of a circular section the Saint-Venant field satisfies \eqref{eq:ratebvp} for any distribution $\tau(r)$, so that only the local equations remain.

\subsection{Variational formulation}
\label{sec:variational}

The equations of the previous subsections can be derived from a single variational equation for dissipative systems \citep{sedov1965mathematical,berdichevsky2009variational}, in the same way as the equations of the single-crystal TDT with non-uniform plastic slip in \citet{le2018dthermodynamic}. The free energy per unit volume is taken as
\begin{equation}
  \psi(\beps^e,\rho,\chi)=\tfrac12K(\tr\beps^e)^2+\mu\,\beps^{e\prime}\!:\!\beps^{e\prime}+e_D\rho-\chi S_c(\rho),\qquad S_c(\rho)=-\rho\ln(a^2\rho)+\rho,
  \label{eq:psi}
\end{equation}
the sum of the elastic energy and of the configurational free energy of the dislocations of \citet{langer2010thermodynamic}: $e_D$ is the energy per unit length of dislocation and $S_c$ the configurational entropy of a random arrangement of density $\rho$ with minimal spacing $a$. In the dimensionless variables the last two terms read $(e_D/a^2)\,\rt\,[1-\ct(1-\ln\rt)]$. The energy functional of the body is
\begin{equation}
  I[\vb{u},\beps^p,\rho,\chi]=\int_\Omega\psi(\sym\grad\vb{u}-\beps^p,\rho,\chi)\,dV.
  \label{eq:I}
\end{equation}
The dissipation potential per unit volume is
\begin{equation}
  D(\dot\beps^p,\dot\rho,\dot\chi)=\Phi(\dot\varepsilon^p_{\eq};\rho,T)+\tfrac12d_\rho\,\dot\rho^2+\tfrac12d_\chi\,\dot\chi^2,\qquad
  \Phi(\dot e;\rho,T)=\int_0^{\dot e}\Sigma(e;\rho,T)\,de,
  \label{eq:D}
\end{equation}
where $\Sigma(\dot\varepsilon^p_{\eq};\rho,T)$ is the flow stress as a function of the plastic strain rate, i.e. the inverse of the depinning law \eqref{eq:q} with respect to $\sigma_{\eq}$. The inverse exists because $q$ is an odd, strictly increasing function of $\sigma_{\eq}$; with the backward term neglected it is explicit,
\begin{equation}
  \Sigma(\dot\varepsilon^p_{\eq};\rho,T)=\mut\sqrt{\rt}\,\Big[\ln\frac{T_P}{T}-\ln\ln\frac{\sqrt{\rt}}{\tz\,\dot\varepsilon^p_{\eq}}\Big]=\mut\sqrt{\rt}\,\zt,
  \label{eq:Sigma}
\end{equation}
which is the expression \eqref{eq:zeta} for $\zt$ read as a function of the rate. Since $\Sigma$ increases with $\dot\varepsilon^p_{\eq}$, $\Phi$ is convex, and $D$ is a convex function of the rates with $D(\vb{0})=0$. The coefficients $d_\rho$, $d_\chi$ are positive functions of the state specified below.

\noindent\textit{Variational equation.} The true fields $\vb{u}$, $\beps^p$, $\rho$, $\chi$ satisfy
\begin{equation}
  \delta I+\int_\Omega\Big(\pdv{D}{\dot\beps^p}\!:\!\delta\beps^p+\pdv{D}{\dot\rho}\,\delta\rho+\pdv{D}{\dot\chi}\,\delta\chi\Big)dV
  =\int_\Omega\rho_m\vb{b}\cdot\delta\vb{u}\,dV+\int_{\partial_t\Omega}\bar{\vb{t}}\cdot\delta\vb{u}\,dS
  \label{eq:varprinciple}
\end{equation}
for all admissible variations: $\delta\vb{u}=\vb{0}$ on $\partial_u\Omega$, $\tr\delta\beps^p=0$, and $\delta\rho$, $\delta\chi$ arbitrary. In the quasi-static case the fields $\vb{u}$ and $\beps^p$ at each instant, together with the rates, are thus determined by the balance of the variation of the energy against the dissipative forces and the external work.

Since $\partial\psi/\partial\beps^e=K(\tr\beps^e)\bI+2\mu\beps^{e\prime}=\bsigma$, the variation with respect to $\vb{u}$ gives $\int_\Omega\bsigma\!:\!\sym\grad\delta\vb{u}\,dV=\int_\Omega\rho_m\vb{b}\cdot\delta\vb{u}\,dV+\int_{\partial_t\Omega}\bar{\vb{t}}\cdot\delta\vb{u}\,dS$, and after integration by parts the equilibrium equation $\div\bsigma+\rho_m\vb{b}=\vb{0}$ in $\Omega$ and the traction condition $\bsigma\vb{n}=\bar{\vb{t}}$ on $\partial_t\Omega$ as natural boundary condition; the spherical relation \eqref{eq:spherical} is contained in the definition of $\bsigma$, because $\tr\beps^p=0$. This is the weak form used in a finite element discretization. The variation with respect to $\beps^p$ gives, on account of $\partial\psi/\partial\beps^p=-\bsigma$ and of the constraint $\tr\delta\beps^p=0$,
\begin{equation}
  \bsigma'=\pdv{\Phi}{\dot\beps^p}=\frac23\,\Sigma(\dot\varepsilon^p_{\eq};\rho,T)\,\frac{\dot\beps^p}{\dot\varepsilon^p_{\eq}},
  \label{eq:flowvar}
\end{equation}
where $\partial\dot\varepsilon^p_{\eq}/\partial\dot\beps^p=\tfrac23\dot\beps^p/\dot\varepsilon^p_{\eq}$ has been used. Taking the von Mises norm of \eqref{eq:flowvar} yields $\sigma_{\eq}=\Sigma(\dot\varepsilon^p_{\eq})$, i.e. the depinning kinetics \eqref{eq:q}, and its direction yields $\dot\beps^p\parallel\bsigma'$, i.e. the Prandtl--Reuss rule \eqref{eq:flowrule}. The flow rule of Section~\ref{sec:flow} is thus the statement that the stress deviator is the derivative of the dissipation potential, and the maximum dissipation property of Proposition~2 is the convexity of $\Phi$. Finally, the variations with respect to $\rho$ and $\chi$ give
\begin{equation}
  \chi\ln\frac{\rho}{\rho_s(\chi)}+d_\rho\dot\rho=0,\qquad -S_c(\rho)+d_\chi\dot\chi=0,
  \label{eq:rhochivar}
\end{equation}
where $e_D+\chi\ln(a^2\rho)=\chi\ln(\rho/\rho_s)$ with $\rho_s=a^{-2}e^{-e_D/\chi}$ has been used. As in \citet{le2018dthermodynamic}, the coefficients are chosen so that \eqref{eq:rhochivar} coincide with the evolution equations \eqref{eq:rhodot2} and \eqref{eq:chidot}:
\begin{equation}
  d_\rho=\frac{\chi\ln(\rho_s/\rho)}{\kappa_\rho\,\dfrac{\bar\mu_T^2\rho\,\dot\varepsilon^p_{\eq}}{\mu_T\,\sigma_{\eq}}\Big(1-\dfrac{\rho}{\rho_s}\Big)},\qquad
  d_\chi=\frac{S_c(\rho)}{\dfrac{\kappa_\chi(T)}{\mu_T}\,\sigma_{\eq}\dot\varepsilon^p_{\eq}\Big(1-\dfrac{\chi}{\chi_0}\Big)}.
  \label{eq:dcoeff}
\end{equation}
Both are positive in the regime in which the theory is used: $\ln(\rho_s/\rho)$ and $1-\rho/\rho_s$ have the same sign for any $\rho$, and $S_c(\rho)>0$ for $a^2\rho<e$, so that $d_\chi>0$ requires $\chi<\chi_0$; this holds for the initial states of Section~\ref{sec:identification}, which lie below the steady state, and then for all times, since $\chi$ increases monotonically towards $\chi_0$ and never crosses it. The variational treatment of $\rho$ and $\chi$ is a formal one, the physics residing in the kinetic coefficients; its content is that the steady state $\rho_s(\chi)$ is the minimizer of the configurational free energy \eqref{eq:psi} at fixed $\chi$, and that the approach to it is dissipative. For $\vb{u}$ and $\beps^p$, by contrast, \eqref{eq:varprinciple} is substantive: it delivers the field equations, the natural boundary conditions and the flow rule from one functional and one potential, and it is the starting point for the weak formulations on which finite element solutions of non-uniform problems in TDT are built.

\section{Reduction to uniaxial tension/compression}
\label{sec:uniaxial}

For $\bsigma=\operatorname{diag}(0,0,\sigma)$, $\dot\beps=\operatorname{diag}(-\nu_t,-\nu_t,1)\dot\varepsilon$ with the appropriate transverse ratio, one has $\sigma_{\eq}=|\sigma|$, $\dot\varepsilon^p_{\eq}=|\dot\varepsilon^p|$, and \eqref{eq:stressrate} reduces to $\dot\sigma=2\mu(1+\nu)(\dot\varepsilon-\dot\varepsilon^p)$ after elimination of the transverse strains by the traction-free condition, exactly as in \citet{le2020averaging}. For constant $\dot\varepsilon$, dividing \eqref{eq:system} by $\dot\varepsilon$ and replacing $\zt(\rt,\qt,\tht)$ by its post-yield value $\zt(\rt,\qt_0,\tht)$, the system coincides with Eq.~(9) of \citet{le2020averaging} with the two replacements
\begin{equation}
  \frac{\pi}{2}\to\sqrt{\tfrac53},\qquad \frac{3\pi}{2}\to\sqrt{15},
  \label{eq:replacements}
\end{equation}
in the relation between $\dot\varepsilon^pt_0$ and $q$ and in the relation between $\mut$ and $\mu_T$, respectively. Both replacements amount to the same factor $(\pi/2)\sqrt{3/5}=1.2167$, and neither affects the dimensionless system: the fitted quantities of \citet{le2020averaging} are $s=\mut/\mu$ and $\qt_0=\dot\varepsilon\tz$, i.e. the combinations $\mut$ and $\tz$ themselves, not $\mu_T$ and $t_0$. Consequently the predictions of the two theories in uniaxial tension/compression are identical with the same parameters and the same initial conditions, and the comparison with the experiments of \citet{johnson1983constitutive} for OFHC copper, ARMCO iron and 4340 steel in \citet{le2020averaging} carries over without change. What changes is only the identification of the microscopic quantities behind the fitted ones,
\begin{equation}
  \mu_T=\frac{a}{b}\frac{s\mu}{\sqrt{15}}\ \text{instead of}\ \frac{a}{b}\frac{2s\mu}{3\pi},\qquad
  t_0=\sqrt{\tfrac53}\,\frac{b}{a}\,\tz\ \text{instead of}\ \frac{\pi b}{2a}\,\tz,
  \label{eq:micro}
\end{equation}
both larger by the factor 1.2167 than in \citet{le2020averaging}.

\section{Torsion of bars and tubes: the torque--twist relation}
\label{sec:torsion}

Torsion tests deliver a torque $M(t)$ and a twist angle $\varphi(t)$. The shear stress--strain curves reported in the literature are reductions of these two records under assumptions (thin wall, or the Nadai construction for solid bars) that presuppose a rate-independent, single-valued $\tau(\gamma)$. Since the theory is rate dependent and the stress varies across the section, we derive the torque--twist relation from the theory, as in \citet{le2018bthermodynamic}, and identify the parameters from the measured $(M,\varphi)$ records directly.

\noindent\textit{Kinematics and statics.} Consider a circular tube $R_i\le r\le R_o$ ($R_i=0$ for a solid bar) of gauge length $L$, twisted by the angle $\varphi(t)$; $\kappa=\varphi/L$ is the twist per unit length. For a circular section the Saint-Venant field $u_\theta=\kappa zr$ involves no warping, so the only strain component is $\varepsilon_{z\theta}=\kappa r/2$, i.e.
\begin{equation}
  \gamma(r,t)=\kappa(t)\,r,\qquad \gdot=\dot\kappa\,r.
  \label{eq:kinematics}
\end{equation}
The stress is zero initially and the flow rule \eqref{eq:flowrule} keeps $\dot\beps^p$ parallel to $\bsigma'$; hence the stress remains a pure shear $\tau(r,t)=\sigma_{z\theta}$ for all times, with $\sigma_{\eq}=\sqrt3|\tau|$ and $\dot\varepsilon^p_{\eq}=\gdot^p/\sqrt3$. To the order stated in Section~\ref{sec:finite} the theory therefore predicts no axial (Swift) effect, and free-end and fixed-end torsion coincide. Equilibrium is satisfied by any distribution $\tau(r)$, and the torque is
\begin{equation}
  M(t)=2\pi\int_{R_i}^{R_o}\tau(r,t)\,r^2\,dr.
  \label{eq:torque}
\end{equation}

\noindent\textit{Local equations.} With $\sigma_{\eq}=\sqrt3|\tau|$ and $\zt=\sqrt3\tau/(\mut\sqrt{\rt})$ the system \eqref{eq:system} reduces at each radius to
\begin{equation}
  \begin{aligned}
    \dot\tau &=\mu\Big[\dot\kappa\,r-\sqrt3\,\frac{\qt(\sqrt3\tau,\rt,\tht)}{\tz}\Big],\\
    \dot\rt &=K_\rho\,\frac{\mut\,\rt\,\qt}{\sqrt3\,\tau\,\tz}\Big[1-\frac{\rt}{\rt_s(\ct)}\Big],\qquad
    \dot\ct =K_\chi(T)\,\frac{\sqrt3\,\tau\,\qt}{\mut\,\tz}\Big(1-\frac{\ct}{\ct_0}\Big),
  \end{aligned}
  \label{eq:local}
\end{equation}
with $\qt$ from \eqref{eq:qtilde}, supplemented by \eqref{eq:adiabatic} at each radius for the dynamic tests (the heating is applied station by station; conduction across the thin wall is neglected). Together with \eqref{eq:torque}, \eqref{eq:local} is the torque--twist relation. In terms of the bare microscopic time the plastic shear rate reads $\gdot^pt_0=\sqrt5\,q(\sqrt3\tau,\rho,T)$ and the mean resolved shear stress is $\bar\tau=\tau/\sqrt5$; these are the torsional counterparts of $\dot\varepsilon^pt_0=\sqrt{5/3}\,q$ and $\bar\tau=\sigma/\sqrt{15}$ in tension, and both follow from the identity $\avg{\tau_{\sm}^2}=J_2/5$ without any further hypothesis. (In \citep{le2018bthermodynamic} the corresponding factors for polycrystalline bars had to be introduced by a separate argument.) In the elastic range \eqref{eq:torque}--\eqref{eq:local} give $M=\mu J\kappa$ with $J=\pi(R_o^4-R_i^4)/2$.

\noindent\textit{Similarity and size.} Writing $r=R_o\xi$, the local system depends on the geometry only through $\dot\kappa R_o$, and \eqref{eq:torque} becomes $M=2\pi R_o^3\int\tau\,\xi^2d\xi$: at equal surface strain $\kappa R_o$ and surface strain rate $\dot\kappa R_o$, $M/R_o^3$ is independent of the size. The present theory thus predicts no size effect in torsion; the size effect measured in \citet{le2019bthermodynamic} requires the non-redundant dislocations of the next step, and the present prediction is the baseline against which it is measured.

\noindent\textit{Numerical solution and the classical reductions.} Equations \eqref{eq:local} are integrated at $n$ Gauss--Legendre stations across the wall ($n=8$ suffices) and \eqref{eq:torque} is evaluated by the same quadrature; the system has $4n$ unknowns (including the station temperatures) and is integrated in a fraction of a second with a stiff solver. For the specimen of \citet{lindholm1980large} used by \citet{johnson1983constitutive} (gauge length $L=3.18$\,mm, wall thickness $t=0.8$\,mm, mean radius $R_m=6.9$\,mm, $t/R_m\approx0.12$) the stress $M/2\pi R_m^2t$ recovered from the computed torque differs from the local solution $\tau(R_m)$ by less than 0.3\,\% up to $\gamma=1$, so the thin-wall reduction used by Johnson and Cook is admissible for this specimen. For a solid bar of the same outer radius, by contrast, the elastic formula $2M/\pi R^3$ overestimates the surface stress by about 25\,\% at $\gamma_R=0.5$--$1$ with the parameters of Table~\ref{tab:parameters}, and the Nadai construction presupposes a single-valued $\tau(\gamma)$ that a rate-dependent theory does not provide. The identification below is therefore formulated in terms of torque and twist, which are the measured quantities; for the thin-walled tube it coincides numerically with an identification in terms of the reduced stresses, for solid bars it does not.

\section{Joint identification from tension, compression and torsion}
\label{sec:identification}

\subsection{Data}

The theory has the same material parameters for every stress state and temperature, so that all available tests on pure copper in the dynamic regime can be fitted together. Quasi-static tests are excluded: as emphasized in \citet{langer2010thermodynamic}, the depinning kinetics describes the driven, rate-dependent regime and not the slow regime in which recovery and recrystallization compete with depinning. We take $\edot,\gdot\ge1$\,s$^{-1}$ as the admissible range. Five sets are used, 142 points in all:
\begin{enumerate}
\item \citet{johnson1983constitutive}, Hopkinson-bar tension of OFHC copper at $298$, $496$ and $732$\,K, $\edot\approx450$\,s$^{-1}$, $\varepsilon\le0.3$ (27 points; the data used in \citet{le2020averaging});
\item \citet{johnson1983constitutive}, torsion of the thin-walled specimen of \citet{lindholm1980large} at $298$\,K at the nominal rates $\gdot=9.3$, $35$ and $148$\,s$^{-1}$: the complete record at $148$\,s$^{-1}$ of their Fig.~2 from $\gamma=0.15$, where the published record begins, to $3$, and the shear stresses at $\gamma=0.2$ and $0.5$ of their Fig.~4 for $9.3$ and $35$\,s$^{-1}$ (22 points), used as torque and twist. The specimen (Fig.~6 of \citet{lindholm1980large}) has an effective gauge length $L=3.18$\,mm, inner diameter $13.0$\,mm and outer diameter $14.6$\,mm, i.e. $R_m=6.9$\,mm and $t=0.8$\,mm, and the reported stress and strain are $\tau=M/(AR_m)$, $\gamma=R_m\varphi/L$ with $A=2\pi R_mt$; these relations are inverted to recover $M$ and $\varphi$. The test at $0.92$\,s$^{-1}$ is excluded as near-quasi-static, the record beyond $\gamma=3$ because it enters the regime of stage-IV hardening, which the theory does not cover;
\item \citet{lindholm1980large}, torsion of the same specimen, OFHC copper CDA 101 annealed at $370^\circ$C for 60 min (grain size $25$--$35\,\mu$m), at $\gdot=174$, $260$ and $330$\,s$^{-1}$: the records of their Fig.~7 from $\gamma=0.25$ to $3$, before the thermal softening and localization that set in at $\gamma\approx3$--$5$ (48 points), converted to torque and twist in the same way. Two limitations of these records are stated in \citet{lindholm1980large} and matter here: at these rates the machine runs open-loop and ``an initial acceleration period exists before constant angular velocity is achieved'', so that the curves below $\gamma\approx1$ are not constant-rate responses and are ordered inversely to the nominal rate; and the $174$ and $260$\,s$^{-1}$ curves coincide within the line width from $\gamma\approx1.3$ on, so that their assignment below that strain rests on the ordering just mentioned. The quasi-static records of the same figure are not used;
\item \citet{samanta1971dynamic}, compression of annealed copper at $298$\,K at $\edot\approx1800$\,s$^{-1}$ in the modified Hopkinson bar, $\varepsilon=0.06$ to $0.52$ (6 points);
\item \citet{samanta1971dynamic}, compression at $873$, $1023$ and $1173$\,K at $\edot=960$ and $2300$\,s$^{-1}$, the complete curves from $\varepsilon\approx0.07$ to $0.8$ (39 points). The first one or two points of each curve belong to the loading transient of the modified Hopkinson bar (Samanta discards the first $80\,\mu$s of every test) and are retained only for completeness.
\end{enumerate}
The three coppers (OFHC of \citet{johnson1983constitutive}, CDA 101 of \citet{lindholm1980large}, 99.9\,\% copper annealed at $650^\circ$C of \citet{samanta1971dynamic}) are assumed to share the material parameters and to differ only in their initial dislocation state. All data were digitized from the published figures.

\subsection{Initial states}

The initial state $(\rt_i,\ct_i)$ of a specimen is not a material parameter but a property of the specimen and its history, and it must be identified together with the material parameters. Specimens of the same batch and heat treatment tested at different rates share their initial state; specimens heated to different temperatures, or from different stock, do not. We therefore assign one initial state to the Johnson--Cook tension specimens at each temperature, one to the Johnson--Cook torsion tubes (whose heat treatment is not documented in \citet{johnson1983constitutive}), one to the Lindholm tubes, and one to Samanta's specimens at each temperature (each set was heated and held before testing): nine states in all. Two constraints follow from the theory. Since $\rt$ and $\ct$ evolve monotonically toward $\rt_s(\ct)$ and $\ct_0$, an initial state above the steady state would produce a softening stress--strain curve, which is not observed in copper; we therefore parametrize
\begin{equation}
  \ct_i=g\,\ct_0,\qquad \rt_i=f\,\rt_s(\ct_i)=f\,e^{-1/\ct_i},\qquad 0<f<1,\quad 0<g<1,
  \label{eq:fg}
\end{equation}
which makes every computed curve monotonically increasing by construction (apart from the small adiabatic softening). Second, a curve with as many free initial parameters as data points is fitted for any material parameters and carries no information on them; this is why the specimens of one batch tested at different rates share one state.

\subsection{Method}

The material parameters $\ct_0$, $K_\rho$, $c_0$, $c_1$ and the nine initial states $(f,g)$ are identified by the least-squares method of \citet{le2020averaging}: the sum of squares of the relative differences between theoretical and measured stresses (sets 1, 4, 5) and torques (sets 2, 3),
\begin{equation}
  F=\sum_{\text{tests}}\frac{1}{N_k}\sum_{j=1}^{N_k}\Big(\frac{y^{th}_{kj}-y^{exp}_{kj}}{y^{exp}_{kj}}\Big)^2,
  \label{eq:F}
\end{equation}
is minimized in the space of the 22 unknowns, each test being weighted by the inverse of its number of points $N_k$ so that the tension curves do not outvote the torques and the hot steady states. The tension and compression curves are computed from \eqref{eq:system} at constant strain rate with the transverse strains eliminated by the traction-free condition, the torque from \eqref{eq:torque}--\eqref{eq:local}, all with adiabatic heating \eqref{eq:adiabatic}. The minimization uses a global search (multistart sequential quadratic programming from scattered trial points, in scaled variables and within the bounds $0.15\le\ct_0\le0.25$, $1\le K_\rho\le10$, $f\le0.9$, $g\le0.98$), iterated with the identified parameters as the new starting point until $F$ stops changing, followed by a final trust-region least-squares polish on the residual vector. The 16 integrations of one evaluation of $F$ (one per test) are independent and run in parallel; one evaluation takes about one second on a laptop with 14 workers, and the whole identification less than an hour.

\subsection{Result}

\begin{table}[t]
\centering\small
\caption{Material parameters for copper. Fixed by the scaling law: $s_T=6.3915\times10^{-3}$, $T_P=45\,000$\,K, $\edot_r=3.16\times10^{12}$\,s$^{-1}$. Identified: $\ct_0$, $c_0$, $c_1$, $K_\rho$. Derived: $s$, $\tz$, $\rt_s$, $K_\chi(T)$. Second row: the compression-only identification of \citet{le2020two}.}
\label{tab:parameters}
\begin{tabular}{lcccccccc}
\toprule
 & $\ct_0$ & $c_0$ & $c_1$ (K) & $K_\rho$ & $s$ & $\tz$ (s) & $\rt_s$ & $K_\chi$ at 298 / 1173 K\\
\midrule
this work & 0.216 & 0.073 & 118 & 1.18 & 0.065 & $3.1\times10^{-14}$ & 0.0097 & 0.91 / 1520\\
\citet{le2020two} & 0.233 & 0.42 & 168 & 1.67 & 0.055 & $3.7\times10^{-14}$ & 0.010 & 2.5 / 450\\
\bottomrule
\end{tabular}
\end{table}

\begin{table}[t]
\centering\small
\caption{Identified initial states of the nine specimen groups, as the ratios $f=\rt_i/\rt_s(\ct_i)$, $g=\ct_i/\ct_0$ of \eqref{eq:fg} and as $(\rt_i,\ct_i)$.}
\label{tab:states}
\begin{tabular}{llcccc}
\toprule
specimens & $T$ (K) & $f$ & $g$ & $\rt_i$ & $\ct_i$\\
\midrule
JC tension \citep{johnson1983constitutive} & 298 & 0.076 & 0.797 & $2.3\times10^{-4}$ & 0.172\\
 & 496 & 0.153 & 0.817 & $5.3\times10^{-4}$ & 0.176\\
 & 732 & 0.716 & 0.685 & $8.3\times10^{-4}$ & 0.148\\
JC torsion tubes \citep{johnson1983constitutive} & 298 & 0.252 & 0.815 & $8.6\times10^{-4}$ & 0.176\\
Lindholm tubes \citep{lindholm1980large} & 298 & 0.112 & 0.743 & $2.2\times10^{-4}$ & 0.160\\
Samanta compression \citep{samanta1971dynamic} & 298 & 0.007 & 0.856 & $3.3\times10^{-5}$ & 0.185\\
 & 873 & 0.076 & 0.497 & $6.8\times10^{-6}$ & 0.107\\
 & 1023 & 0.110 & 0.542 & $2.1\times10^{-5}$ & 0.117\\
 & 1173 & 0.091 & 0.737 & $1.7\times10^{-4}$ & 0.159\\
\bottomrule
\end{tabular}
\end{table}

Tables~\ref{tab:parameters} and \ref{tab:states} and Fig.~\ref{fig:identification} give the result. The rms relative errors are $5.3$\,\% for the 27 Johnson--Cook tension points, $5.7$\,\% for their 22 torques, $7.1$\,\% for the 48 Lindholm torques, $3.1$\,\% for Samanta's room-temperature compression and $4.9$\,\% for his 39 hot compression points, $5.9$\,\% over all 142 points, which span $298$--$1173$\,K, $10$--$2300$\,s$^{-1}$ and three stress states. All computed curves are monotone. Several remarks.

(i) \emph{Consistency with the scaling-law identification.} The material parameters come out where the compression-only identification of \citet{le2020two} had put them: $\ct_0=0.216$ against $0.233$, $K_\rho=1.18$ against $1.67$, and the derived $s=0.065$ and $\tz=3.1\times10^{-14}$\,s against $\tilde r=0.055$ and $3.7\times10^{-14}$\,s, now with tension and torsion added to the data set. The one quantity that changes is the temperature dependence of $K_\chi$: $c_0=0.073$, $c_1=118$\,K give $K_\chi=0.91$ at $298$\,K and $120$ at $873$\,K against $2.5$ and $75$ in \citet{le2020two}. With the complete hot compression curves in the fit the exponential is now decided by their transients, and it is steeper.

(ii) \emph{The parameters $\ct_0$, $K_\rho$, $c_0$ compensate one another.} The least-squares surface is flat along a valley in these three: with the bounds released, minima of practically the same $F$ exist at $\ct_0$ down to $0.11$ with $K_\rho$ at $0.1$ and $c_0$ ten times larger. The bounds $\ct_0\in[0.15,0.25]$ and $K_\rho\in[1,10]$ are therefore not a numerical convenience but the physical input that selects the minimum: the first is Langer's bound on the steady-state configurational temperature \citep{langer2010thermodynamic}, both are the neighborhood of the values of \citet{le2020two}. That the minimum then settles at the values of \citet{le2020two} rather than at a bound is the consistency check.

(iii) \emph{The torsion tubes.} With the Johnson--Cook torsion specimens sharing the initial state of the room-temperature tension specimens, their torques misfit by $12$\,\% rms and the misfit is systematic: the measured torque at $\gamma=0.2$ corresponds to $\sqrt3\tau=291$\,MPa at $\varepsilon_{\eq}=0.115$ and $148$\,s$^{-1}$, whereas tension at $451$\,s$^{-1}$ gives $215$\,MPa at the same strain, while at $\gamma=0.5$ the two agree. Johnson and Cook noted this tension--torsion discrepancy and averaged their constants over both tests. No $J_2$ theory can fit both with one initial state, and the third-invariant dependence of any partition of slip (Remark~\ref{rem:family}) is far too small, and of the wrong sign, to account for it. With their own initial state the torques at $9.3$ and $35$\,s$^{-1}$ are fitted to $2$\,\% and the complete record at $148$\,s$^{-1}$ to $6$\,\% (see (iv) for its shape), and the identified state, $\rt_i=8.6\times10^{-4}$, is four times that of the tension specimens. The Lindholm tubes, from the same apparatus and specimen design, confirm the picture: their curves reach $\tau=145$--$180$\,MPa, i.e. $\sqrt3\tau\approx250$--$300$\,MPa, at $\gamma=0.25$--$0.5$, three times the yield stress of the Hopkinson-bar specimens. The tubes were annealed at $370^\circ$C for one hour and are fine-grained ($25$--$35\,\mu$m) \citep{lindholm1980large}; the theory, which has no grain-boundary strengthening, represents this harder state through a higher initial dislocation density. Since Johnson and Cook do not document the state of their tubes, a separate initial state is the only admissible treatment, and it is what the data require.

(iv) \emph{Where the theory is at its limits.} Three residual patterns are systematic and should be named. The Johnson--Cook tension curve at $732$\,K is fitted with an S-shaped error ($-9$\,\% at $\varepsilon=0.06$, $+10$\,\% at $0.3$) and an initial state close to saturation ($f=0.72$): at this temperature the exponential \eqref{eq:KchiT} anchored at $298$ and $873$--$1173$\,K gives $K_\chi\approx36$, and the curve would prefer a slower approach to its steady state -- the price of one exponential over $298$--$1173$\,K. The three Lindholm curves are fitted with errors of opposite sign at the two ends of the rate range ($-2$ to $-11$\,\% at $174$\,s$^{-1}$, $+8$ to $+13$\,\% at $330$\,s$^{-1}$), because the measured curves are ordered inversely to the nominal rate over the whole window, a consequence of the acceleration period of the machine that no rate-dependent constitutive law can reproduce; restricted to $\gamma\ge1.25$ their rms error is $5.6$\,\%. The complete Johnson--Cook record at $148$\,s$^{-1}$ is followed to $3$\,\% between $\gamma=0.25$ and $1.5$, but the measured torque keeps rising at a nearly constant rate to $\gamma=3$ (from $50$ to $57$\,N\,m) where the computed one has saturated, so that the error grows to $-12$\,\% at $\gamma=3$; this is the onset of stage-IV hardening, which requires the evolution of texture and of the cell structure and is outside the present theory, as noted in Section~\ref{sec:finite}. The Lindholm curves at the same strains do not show it, their plateau being followed by softening, which is why the two torsion data sets could be fitted with one material law only up to $\gamma=3$. The torque rise between $9.3$ and $148$\,s$^{-1}$ at $\gamma=0.5$, $7$\,\% computed against $12$\,\% measured, is fixed by $s_T$ and $T_P$ of the scaling law and is not adjusted to the torsion data.

\begin{figure}[t]
\centering
\includegraphics[width=\textwidth]{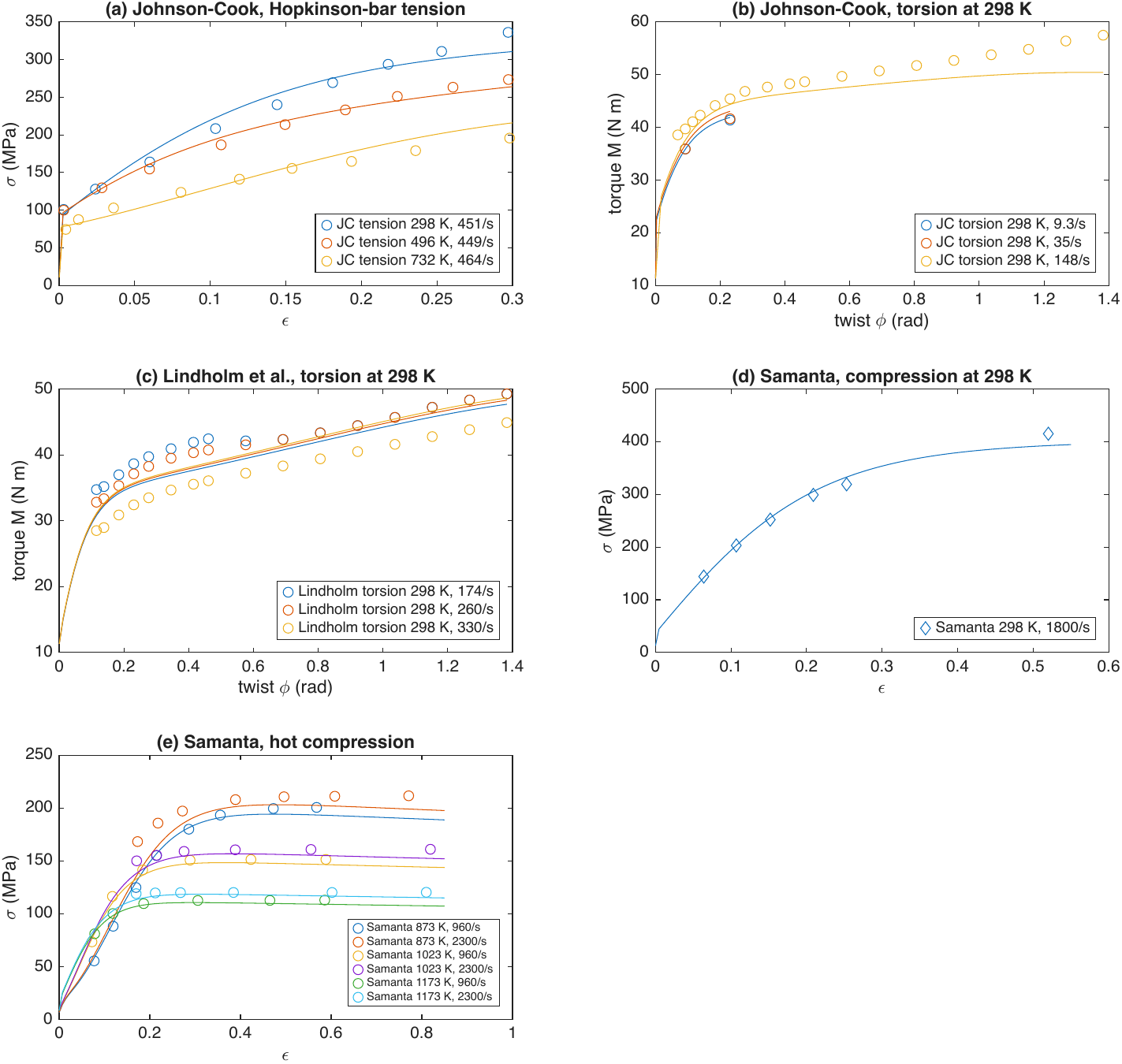}
\caption{Joint identification for copper with the parameters of Table~\ref{tab:parameters}; symbols: experiment, lines: theory. (a) Hopkinson-bar tension at three temperatures \citep{johnson1983constitutive}; (b) torque against twist for the thin-walled tubes at $9.3$, $35$ (two points each) and $148$\,s$^{-1}$ (complete record) \citep{johnson1983constitutive} (model drawn over the range of each record); (c) torque against twist at $174$, $260$ and $330$\,s$^{-1}$ \citep{lindholm1980large}; (d) compression at $298$\,K and $1800$\,s$^{-1}$ \citep{samanta1971dynamic}; (e) compression at $873$, $1023$ and $1173$\,K at $960$ and $2300$\,s$^{-1}$ \citep{samanta1971dynamic}.}
\label{fig:identification}
\end{figure}

\section{Discussion and outlook}
\label{sec:discussion}

The theory obtained here is an associated $J_2$ flow theory with rate- and temperature-dependent isotropic hardening. Its structure -- the von Mises potential, the Prandtl--Reuss flow direction, and the geometric factors $\sqrt{15}$ and $\sqrt{5/3}$ that connect the macroscopic stress and strain rate to the resolved quantities of a slip system -- follows from two hypotheses, equal probability of grain orientations and linear partition of slip over orientations, and none of it is fitted. The hardening is supplied by thermodynamic dislocation theory through the depinning kinetics and the evolution of $\rho$ and $\chi$, driven by the plastic power $\sigma_{\eq}\dot\varepsilon^p_{\eq}$. Written as rates with respect to time, the equations apply to arbitrary loading paths, and the torque--twist relation of torsion follows without the classical reductions.

The identification shows that the same material parameters -- and, up to the temperature dependence of $K_\chi$, the same values as the earlier identification from compression alone -- describe tension, compression and torsion of copper over more than two decades of dynamic strain rate and from room temperature to two-thirds of the melting temperature, once three things are respected: the quasi-static regime is excluded, the pinning temperature and the reference strain rate are taken from the scaling law for the steady-state flow stress rather than fitted anew, and the initial dislocation state is treated as a property of the specimen. The last point resolves the tension--torsion discrepancy of Johnson and Cook in favor of the material law. It also shows the limit of what a least-squares identification can do: the parameters $\ct_0$, $K_\rho$ and $c_0$ compensate one another in the flow stress, and only physical bounds, or independent information on the initial dislocation density of the annealed material, select among the equivalent minima.

Three extensions are indicated. First, texture: for a textured material the measure \eqref{eq:measure} must be replaced by the orientation distribution function; then $\bT$ is no longer isotropic, and the same two identities produce an anisotropic quadratic potential of Hill type with the flow direction $\bT$-weighted. Texture evolution would require coupling to the lattice rotation -- the point at which the present macroscopic route and the crystal-plasticity route of \citet{lieou2020thermodynamic} meet -- and it is also the mechanism behind the large-strain torsion curves beyond $\gamma\approx1$ that the present theory, saturating at $\rt_s$, does not follow. Second, torsion of solid bars: the theory predicts the shape of the torque--twist curve of a bar and the ratio of the torques of bars of different radii at equal surface strain without any further parameter, and the measured size effect in torsion \citep{le2019bthermodynamic} is then the signature of the geometrically necessary dislocations that the next step introduces. Third, kinematic hardening and the Bauschinger effect are absent here because the internal state is described by the scalars $\rho$ and $\chi$ of redundant dislocations; they arise naturally from non-redundant dislocations in non-uniform macroscopic deformation, which is the subject of that step.

\bibliographystyle{elsarticle-harv}
\bibliography{refs}

\end{document}